\documentclass[conference]{IEEEtran}

\usepackage{amssymb}
\usepackage{amsmath}
\usepackage{graphicx}
\usepackage{cite}
\usepackage{citesort}
\usepackage{subfigure}
\usepackage{graphicx,epstopdf}
\usepackage{epsfig}	
\usepackage{cite,graphicx,amsmath,amssymb}
\usepackage{comment}
\usepackage{lipsum}

\newtheorem{theorem}{\textbf{Theorem}}

\newtheorem{lemma}{\textbf{Lemma}}

\newtheorem{corollary}{\textbf{Corollary}}

\makeatletter
\def\ScaleIfNeeded{%
\ifdim\Gin@nat@width>\linewidth \linewidth \else \Gin@nat@width
\fi } \makeatother

\begin{document}
%

\title{\Huge{Millimeter Wave Power Transfer and Information Transmission}}

\author{
\IEEEauthorblockN{ Lifeng Wang\IEEEauthorrefmark{1}, Maged Elkashlan\IEEEauthorrefmark{1}, Robert W. Heath, Jr.\IEEEauthorrefmark{2}, Marco Di Renzo\IEEEauthorrefmark{3}, and Kai-Kit Wong\IEEEauthorrefmark{4}} \IEEEauthorblockA{
\IEEEauthorrefmark{1}School of Electronic Engineering and Computer
Science, Queen Mary University of London, London, UK\\
\IEEEauthorrefmark{2} Department of Electrical and Computer Engineering, The University of Texas at Austin,
Austin, Texas, USA\\
\IEEEauthorrefmark{3} CNRS-CENTRAL/SUPELEC-University Paris-
Sud XI, 3 rue Joliot-Curie, 91192 Gif-sur-Yvette, France\\
\IEEEauthorrefmark{4} Department of Electronic and
Electrical Engineering, University College London, London, UK} }

\maketitle

\begin{abstract}
 Compared to the existing lower frequency wireless power transfer, millimeter wave (mmWave) power transfer takes advantage of the high-dimensional multi-antenna and narrow beam transmission. In this paper we introduce wireless power transfer for mmWave cellular networks. Here, we consider users with large energy storage that are recharged by the mmWave base stations prior to uplink information transmission, and analyze the average harvested energy and average achievable rate. Numerical results corroborate our analysis and show that the serving base station plays a dominant role in wireless power transfer, and the contribution of the interference power from the interfering base stations is negligible, even when the interfering base stations are dense. By examining the average achievable rate in the uplink, when increasing the base station density, a transition from a noise-limited regime to an interference-limited regime is observed.

\end{abstract}

\section{Introduction}
Triggered by the recent development of efficient rectifier circuit design and wireless networks with multi-antenna techniques such as beamforming, wireless power transfer has regained attention. Compared to traditional energy sources such as solar, wind, and thermoelectric, the advantages of wireless power transfer using radio frequency (RF) energy are at least two-fold:  1) it is independent of the environment and can be applied in any places; and 2) it is flexible and can be scheduled at any time.

The existing literature studies wireless power transfer that operates in the networks using lower frequency bands. In \cite{Kaibin2014}, the deployment of power beacons for powering a cellular network was investigated based on a stochastic geometry network model. In \cite{Seunghyun2013}, secondary transmitters in cognitive radio were proposed to harvest energy from the primary transmissions as well as reuse the spectrum from the primary transmitters. The impacts of network density and power splitting RF harvesting on the outage performance and the harvested energy was analyzed in \cite{Krikidis_SP_RF_2014}. The performance of a wireless sensor powered by ambient RF energy was presented in \cite{Flint_Globecom}, where the RF energy sources were assumed to be located following a Ginibre $\alpha$-determinantal point process.

A limitation of prior lower frequency work is that power transfer may affect the quality of service of the existing cellular networks, since the surrounding energy signals from the energy sources such as power beacons~\cite{Kaibin2014} may result in large interference inflicted on the information receiver. In addition, the current cellular spectrum is heavily utilized. The millimeter wave (mmWave) spectrum is emerging as the new mobile broadband~\cite{Zhouyue2011_magazine,TED2013IEEE_Access}, which provides large bandwidths for ultra-high data rates.  MmWave spectrum is also a promising solution for the large-scale in-band backhaul~\cite{Sooyoung2013}. With this in mind, mmWave transmission is a key enabler in future fifth generation (5G) cellular networks.

We believe that  mmWave is a promising and potentially highly rewarding candidate for wireless power transfer, due to the following factors:
 \begin{itemize}
   \item \textbf{High frequencies} Unlike the current wireless power transfer in lower frequencies, mmWave power transfer has no impact on the existing cellular transmissions, since it operates in higher frequencies.

   \item \textbf{Narrow beams} In mmWave systems, the narrow beams or directed beams are typically used, which can exploit directivity gains.

   \item \textbf{Large array gains} Due to the shorter wavelengths, large antenna arrays can be easily deployed in mmWave systems, which can bring large array gains.

    \item \textbf{Dense networks} In future networks, mmWave base stations (BSs) will be densely deployed. As such, the distance between the user and the serving BS is shorter compared to that of existing cellular networks, which decreases the path loss.

 \end{itemize}
In fact, some rectifier circuit designs for mmWave power
transfer have been proposed in the literature such as~\cite{Yu_Jiun2007,Shinohara2011}. More recently,~\cite{Ladan2013} proposed a dual diode  rectifier circuit operating at K-band, which can achieve 40$\%$ RF-to-DC conversion efficiency driven by a 35 mW input power.

In this paper, we propose wireless power transfer in mmWave cellular networks. The user harvests energy from the serving BS, then uses the harvested energy to transmit information messages. We consider directional transmission and reception, with analog processing and phase shifters. Stochastic geometry is employed to model the positions of the BSs with blockage effects as in \cite{Tianyang_arxiv2014}. Our results demonstrate that the interference power received at the user has little contribution to the amount of energy harvested. In the uplink, we find that
 mmWave power transfer networks are noise-limited when the BSs are not super dense. Increasing the density of the BSs beyond a critical point, however, will shift the network into an interference-limited regime.

\section{System Description}\label{system_UA}
We consider a mmWave time-division duplex (TDD) cellular network, where users are powered by the outdoor mmWave BSs via mmWave power transfer prior to information transmission in the uplink. To eliminate the BS-to-BS and user-to-user interference, transmissions in different cells are synchronous\footnote{The synchronous operation with a common uplink-downlink configuration in multiple cells is typically used in the current TDD cellular networks~\cite{Zukang2012}.}.  The locations of the BSs follow a homogeneous Poisson point process
 (PPP) $\Phi$ with density $\rho$. The users are located following a homogeneous PPP $\Psi$ with density $\mathcal{U}$. It is assumed that the user density is much larger than the BS density such that there exists one active mobile at each time slot
in each cell. Each BS is equipped with $M$-element antenna array and each user is equipped with $N$-element antenna array.
 In light of the blockage effects in the outdoor scenario, a user is associated with either a
 line-of-sight (LoS) mmWave BS or a non-line-of-sight (NLoS) mmWave BS~\cite{Tianyang2014_Sept_com}.
 Let $\Phi_\mathrm{LoS}$ be the point process of LoS BSs  and $\Phi_\mathrm{NLoS}=\Phi\setminus \Phi_\mathrm{LoS}$ be
 the point process of NLoS BSs.
 We denote $f_\mathrm{Pr}\left(R\right)$ as the probability that a link at a distance $R$ is LoS, while the NLoS probability of a link
 is $1-f_\mathrm{Pr}\left(R\right)$.  The LoS probability function $f_\mathrm{Pr}\left(R\right)$ can be obtained from
 field measurements or stochastic blockage models\footnote{When the blockages are modeled as a rectangle boolean scheme, $f_\mathrm{Pr}\left(R\right)=e^{-\varrho R}$, where $\varrho$ is a parameter determined by
 the density and the average size of the blockages~\cite{Tianyang_arxiv2014}.}, as mentioned in \cite{Tianyang_arxiv2014}.
 Each user is connected to the BS which has the smallest path loss\footnote{Since each BS has the same transmit power and uses the same transmission scheme, user association based on the smallest path loss is equivalent to that based on the maximum receive power.}, and users within the same cell are scheduled based on time-division multiple-access (TDMA).

The sparse scattering mmWave environment makes many traditional fading distributions inaccurate for the modeling of
the mmWave channel~\cite{Ayach2014}. In mmWave transmission, the strongest physical path occupies the dominant role, which indicates that the link budget is cut by transmitting and
receiving along the strongest path~\cite{Zhouyue_2012_asilomar}. Moreover, with the application of highly directional antennas, large antenna arrays and
 broadband signal transmission, the impact of small-scale fading on the received signal power is negligible~\cite{TED2013IEEE_Access}. As such,
 we focus on the dominant propagation path. Accordingly, the channel model is established
 as $ {\bf{H}} = \sqrt{ L\left(R\right)} {\bf{a}}_\mathrm{r}\left(\theta_\mathrm{\mathrm{r}}\right){{\bf{a}}_\mathrm{t}^\dag}\left(\theta_\mathrm{t}\right)$,
where $L\left(R\right)$ is the path loss function, $\theta_\mathrm{r}$ is the angle of arrival (AoA) and $\theta_\mathrm{t}$ is the angle of
departure (AoD), ${\bf{a}}_\mathrm{r}\left(\theta_\mathrm{r}\right)$ is the receive array response vector, and ${\bf{a}}_\mathrm{t}\left(\theta_\mathrm{t}\right)$ is
 the transmit array response vector. For an $\ell$-element uniform linear array (ULA), its corresponding array response vector is written as ${\bf{a}}_\mathrm{ULA}\left(\theta\right)={\left[ {1,{e^{ - j2\pi \frac{d}{\lambda }\sin \left( {{\theta}} \right)}}, \cdots ,{e^{ - j2\pi \left( {\ell - 1} \right)\frac{d}{\lambda }\sin \left( {{\theta}} \right)}}} \right]^T}\in {\mathcal{C}^{\ell\times 1}}$ at the direction $\theta$, where $\lambda $ is the wavelength of propagation, and $d$ is the antenna spacing\footnote{The array response vector for uniform planar array can be seen in~\cite{Ayach2014}.}.
We consider two different path loss laws: $L\left(R\right)=\beta_\mathrm{LoS} R^{-\alpha_\mathrm{LoS}}$ is the path loss law for LoS channel and $L\left(R\right)=\beta_\mathrm{NLoS} R^{-\alpha_\mathrm{NLoS}}$ is the path loss law for NLoS channel, where $\beta_\mathrm{LoS}$, $\beta_\mathrm{NLoS}$ are the frequency dependent constant values and $\alpha_\mathrm{LoS}$, $\alpha_\mathrm{NLoS}$ are the path loss exponents~\cite{Tianyang2014_Sept_com}.

Due to the high cost of mmWave RF chains and power consumption, the low-cost low-complexity analog beamforming is an appealing approach in mmWave transmission. Hence we adopt analog beamforming with phase shifters~\cite{Doan2004}. Particularly,
the low-complexity matched filter (MF) is adopted at the mmWave BSs and users.

\subsection{Power Transfer Model}
In the power transfer phase, users are powered by the BSs. We assume that a typical user is located at the origin $o$. Let $\theta_{\mathrm{t}_o}$  and $\theta_{\mathrm{r}_o}$ denote the AoD and AoA, respectively, the serving BS can orient its transmit beam along $\theta_{\mathrm{t}_o}$ by using MF ${\bf{w}}_{\mathrm{t}_o}=\frac{1}{\sqrt M}{\bf{a}}_\mathrm{t}\left( {{\theta_{\mathrm{t}_o}}} \right)$, and the typical user can orient its receive beam along $\theta_{\mathrm{r}_o}$ by using MF ${\bf{w}}_{\mathrm{r}_o}=\frac{1}{\sqrt N}{\bf{a}}_\mathrm{r}^\dag\left( {{\theta_{\mathrm{r}_o}}} \right)$. Likewise, for a user located along ${\theta_{\mathrm{t}_k}}$ in the $k$-th cell, its serving BS's transmit beam is ${\bf{w}}_{\mathrm{t}_k}=\frac{1}{\sqrt M}{\bf{a}}_\mathrm{t}\left( {{\theta_{\mathrm{t}_k}}} \right)$ while the typical user $o$ is located along ${\vartheta_{{\mathrm{t}_k}}}$  seen by the $k$-th cell BS. We use the short-range propagation model~\cite{Baccelli2006_gen,Kaibin2014} to avoid singularity caused by proximity between the BSs and the users. This ensures that users receive finite average power. Hence, the downlink power transfer channel from the $k$-th BS to the typical user can be modeled
as
\begin{align}\label{PT_Channel_model}
{\bf{H}}_k = \sqrt {L\left( \max\left\{ {{R_k}}, \mathrm{D}\right\}\right)} {{\bf{a}}_{\mathrm{r}}}\left( {{\vartheta _{{\mathrm{r}_k}}}} \right){\bf{a}}_{\mathrm{t}}^\dag\left( {{\vartheta_{{\mathrm{t}_k}}}} \right),
\end{align}
where $\mathrm{D}>0$ is the reference distance and ${\vartheta _{\mathrm{r}_k}},   \vartheta_{{\mathrm{t}_k}}\sim\emph{U}\left(0,2\pi\right)$~\cite{Akdeniz2014}. Since the energy harvested from the receiver's noise is negligible, the receive power at a typical user $o$ is written as
\begin{align}\label{receive_power}
&\hspace{-0.2 cm} {P_{\mathrm{r}_o}}=\frac{{P_{\mathrm{mm}}}}{NM} L\left(\max\left\{  {{R_o}}, \mathrm{D}\right\}\right)\left|{\bf{a}}_\mathrm{r}^\dag\left( {{\theta_{\mathrm{r}_o}}} \right) {\bf{a}}_\mathrm{r}\left( {{\theta_{\mathrm{r}_o}}} \right){\bf{a}}_\mathrm{t}^\dag\left( {{\theta_{\mathrm{t}_o}}} \right) {\bf{a}}_\mathrm{t}\left( {{\theta_{\mathrm{t}_o}}} \right)\right|^2 +\nonumber\\
&\hspace{-0.2 cm} \frac{{{P_{\mathrm{mm}}}}}{{NM}}\sum\limits_{k \in \Phi \setminus \left\{ o \right\} } L\left(\max\left\{{ {{R_k}}}, \mathrm{D}\right\}\right) \left| {\bf{a}}_\mathrm{r}^\dag \left( {{\theta _{\mathrm{r}_o}}} \right){{\bf{a}}_\mathrm{r}}\left( {{\vartheta _{\mathrm{r}_k}}} \right){{\bf{a}}_\mathrm{t}^\dag}\left( {{\vartheta _{\mathrm{t}_k}}} \right){\bf{a}}_\mathrm{t} \left( {{\theta _{\mathrm{t}_k}}} \right)\right|^2\nonumber\\
&\hspace{-0.2 cm} = \underbrace {NM{P_{\mathrm{mm}}}L\left(\max\left\{  {{R_o}}, \mathrm{D}\right\}\right)}_{\mathrm{En}_1}+ \underbrace { \frac{{{P_{\mathrm{mm}}}}}{{NM}}\sum\limits_{k \in \Phi \setminus \left\{ o \right\} } {\hbar _k} L\left( \max\left\{{{{R_k}}}, \mathrm{D}\right\}\right)}_{\mathrm{En}_2},
\end{align}
where  $\mathrm{En}_1$ is the receive power from the serving BS and $\mathrm{En}_2$ is the receive power from the interfering BSs, $P_{\mathrm{mm}}$ is the BS's transmit power, ${\hbar _k}=\frac{{\left( {1 - \cos \left( {N{\varpi _{{\vartheta _{{\mathrm{r}_k}}},{\theta _{{\mathrm{r}_o}}}}}} \right)} \right)
\left( {1 - \cos \left( {M{\varpi _{{\vartheta _{{\mathrm{t}_k}}},{\theta _{{\mathrm{t}_k}}}}}} \right)} \right)}}{{\left( {1 - \cos \left( {{\varpi _{{\vartheta _{{\mathrm{r}_k}}},{\theta _{{\mathrm{r}_o}}}}}} \right)} \right)
\left( {1 - \cos \left( {{\varpi _{{\vartheta _{{\mathrm{t}_k}}},{\theta _{{\mathrm{t}_k}}}}}} \right)} \right)}}$, $\theta _{{\mathrm{t}_k}}\sim\emph{U}\left(0,2\pi\right)$, and
${\varpi _{\vartheta,{\theta}}}  = 2\pi \frac{d}{\lambda }\left( {\sin \left( {{\vartheta}} \right) -
\sin \left( {{\theta}} \right)} \right)$.

\textbf{Remark 1}: From \eqref{receive_power}, we see that in mmWave systems, the user can at least be recharged by the input DC power $\eta NM{P_{\mathrm{mm}}}L\left(\max\left\{  {{R_o}}, \mathrm{D}\right\}\right)$, where $\eta$ is the RF-to-DC conversion efficiency.

\subsection{Uplink Information Transmission}
After energy harvesting,  user $\mathrm{u}_i$ transmits information signals to the serving BS by setting a transmit power value $P_{\mathrm{u}_i}$. The signal-to-interference-plus-noise ratio (SINR) at a typical serving BS is given by
\begin{align}\label{Eq_1_inform}
\mathrm{SINR}=\frac{MN{P_{\mathrm{u}_o}}L\left( {\max\left\{  {{R_o}}, \mathrm{D}\right\}} \right)}{\underbrace{\frac{{1}}{{MN}}\sum\limits_{i \in \widetilde{\mathcal{U}} \setminus \left\{ o \right\} } {P_{\mathrm{u}_i}} L\left( \max\left\{  {{R_i}}, \mathrm{D}\right\}\right)
{\hbar_i}}_{\mathrm{I}_\mathrm{U}}+\delta^2},
\end{align}
where ${\hbar_i}=\frac{{\left( {1 - \cos \left( {M{\varpi _{{\vartheta _{{\mathrm{r}_i}}},{\theta _{{\mathrm{r}_o}}}}}} \right)} \right)\left( {1 - \cos \left( {N{\varpi _{{\vartheta _{{\mathrm{t}_i}}},{\theta _{{\mathrm{t}_i}}}}}} \right)} \right)}}{{\left( {1 - \cos \left( {{\varpi _{{\vartheta _{{\mathrm{r}_i}}},{\theta _{{\mathrm{r}_o}}}}}} \right)} \right)\left( {1 - \cos \left( {{\varpi _{{\vartheta _{{\mathrm{t}_i}}},{\theta _{{\mathrm{t}_i}}}}}} \right)} \right)}}$, $\widetilde{\mathcal{U}}$ is the point process corresponding to the interfering users, $\mathrm{I}_{\mathrm{U}}$ is the uplink interference from the other  users, and $\delta^2$ is the noise power at the typical serving BS.

\begin{figure*}[!t]
\normalsize
 \begin{align}\label{average_power_transfer_1}
& {\overline P}_\mathrm{r_o}= {\overline P}_\mathrm{r_o}^\mathrm{Low}+\frac{{{P_{\mathrm{mm}}}}}{{NM}}\overline{\hbar} \beta_\mathrm{LoS}
\left(2 \pi \rho\right)^2 \bigg\{\Big.\nonumber\\
&+\int_{0}^{\infty}{ {\left[\int_{x}^{\infty}
{ \left(\max\left\{t,\mathrm{D}\right\}\right)^{-\alpha_\mathrm{LoS}}}f_\mathrm{Pr}\left(t\right) t dt \right] }  x{f_{\Pr }}
\left( x \right){e^{ - 2\pi \rho \left[ {\Theta \left( x \right) + \Xi \left( {{\varphi_\mathrm{LoS}}
\left( x \right)} \right)} \right]}} dx}\nonumber\\
&+\int_{0}^{\infty}{ {\left[\int_{\varphi_\mathrm{NLoS}\left(x\right)}^{\infty}
{ \left(\max\left\{t,\mathrm{D}\right\}\right)^{-\alpha_\mathrm{LoS}}}f_\mathrm{Pr}\left(t\right) t dt \right] }x{\left(1-f_{\Pr }\left( x \right)\right)}
{e^{ - 2\pi \rho \left[ \Theta \left(\varphi_\mathrm{NLoS}\left(x\right)\right)+ {\Xi \left( x\right)} \right]}} dx}\bigg\}\nonumber\\
&+\frac{{{P_{\mathrm{mm}}}}}{{NM}}\overline{\hbar} \beta_\mathrm{NLoS}
\left(2 \pi \rho\right)^2 \bigg\{\Big.\nonumber\\
&+\int_{0}^{\infty}{ {\left[\int_{\varphi_\mathrm{LoS}\left(x\right)}^{\infty}
{ \left(\max\left\{t,\mathrm{D}\right\}\right)^{-\alpha_\mathrm{NLoS}}} \left(1-f_\mathrm{Pr}\left(t\right)\right) t dt\right] } x{f_{\Pr }}
\left( x \right){e^{ - 2\pi \rho \left[ {\Theta \left( x \right) + \Xi \left( {{\varphi_\mathrm{LoS}}
\left( x \right)} \right)} \right]}} dx}\nonumber\\
&+ \int_{0}^{\infty}{ {\left[\int_{x}^{\infty}
{ \left(\max\left\{t,\mathrm{D}\right\}\right)^{-\alpha_\mathrm{NLoS}}} \left(1-f_\mathrm{Pr}\left(t\right)\right) t dt\right] }x{\left(1-f_{\Pr }\left( x \right)\right)}
{e^{ - 2\pi \rho \left[ \Theta \left(\varphi_\mathrm{NLoS}\left(x\right)\right)+ {\Xi \left( x\right)} \right]}} dx}\bigg\}
 \end{align}
\hrulefill \vspace*{0pt}
\end{figure*}

\section{Millimeter Wave Powered Networks}\label{mmWave_WPT}
In this section, the average harvested energy is derived, assuming all users are equipped with large energy storage.
We also assume that each frame time
duration is $T$.
 There are two time slots in each frame: 1) In the first time slot,
the user receives the power from the serving BS, which occupies $\phi T$ time, where $\phi$ ($0<\phi<1$) is the fraction of the time; and 2) In the second time slot, the user transmits the information
signals to the serving BS, which occupies $\left(1-\phi\right) T$ time.

Due to large energy storage, users can transmit the signal with reliable transmit power after energy harvesting. Considering the fact that the energy consumed for uplink information transmission should not exceed the harvested energy,
 a stable transmit power up to $\eta\frac{\phi}{\left(1-\phi\right)}\mathbb{E}\left\{  P_\mathrm{r_o}\right\}$ can be provided, as suggested in~\cite{Kaibin2014}. As such, the average receive power
$ {\overline P}_\mathrm{r_o}=\mathbb{E}\left\{  P_\mathrm{r_o}\right\}$ is pivotal in this case.
Therefore, we focus on the average receive power and have the following theorem:
 \begin{theorem}
\emph{ The average receive power} $ {\overline P}_\mathrm{r_o}$ \emph{is derived as \eqref{average_power_transfer_1} at the top of the next page, where $\overline{\hbar}$ is given by \eqref{h_bar_1}, and the lower bound ${\overline P}_\mathrm{r_o}^\mathrm{Low}$ for the average receive power ${\overline P}_\mathrm{r_o}$ is given by}
\begin{align}\label{average_transmit_power}
& \hspace{-0.4cm}{\overline P}_\mathrm{r_o}^\mathrm{Low}= NM{P_{\mathrm{mm}}}{{2\pi \rho }} \bigg\{ \beta_\mathrm{LoS}{\rm{D}}^{-\alpha_\mathrm{LoS}}\nonumber\\
&\hspace{-0.4cm} \times \int_0^{\rm{D}} x{f_{\Pr }}
\left( x \right){e^{ - 2\pi \rho \left[ {\Theta \left( x \right) + \Xi \left( {{\varphi_\mathrm{LoS}}
\left( x \right)} \right)} \right]}} {dx}\nonumber\\
&\hspace{-0.4cm}+\beta_\mathrm{LoS} \int_{\rm{D}}^\infty x^{1-\alpha_\mathrm{LoS}} {f_{\Pr }}
\left( x \right){e^{ - 2\pi \rho \left[ {\Theta \left( x \right) + \Xi \left( {{\varphi_\mathrm{LoS}}
\left( x \right)} \right)} \right]}}  {dx} \nonumber\\
&\hspace{-0.4cm}+\beta_\mathrm{NLoS}{\rm{D}}^{-\alpha_\mathrm{NLoS}} \int_0^{\rm{D}} x{\left(1-f_{\Pr }\left( x \right)\right)}
{e^{ - 2\pi \rho \left[ \Theta \left(\varphi_\mathrm{NLoS}\left(x\right)\right)+ {\Xi \left( x\right)} \right]}} {dx}\nonumber\\ &\hspace{-0.4cm}+\beta_\mathrm{NLoS}\int_{\rm{D}}^\infty x^{1-\alpha_\mathrm{NLoS}} {\left(1-f_{\Pr }\left( x \right)\right)}
{e^{ - 2\pi \rho \left[ \Theta \left(\varphi_\mathrm{NLoS}\left(x\right)\right)+ {\Xi \left( x\right)} \right]}} {dx}\bigg\},
\end{align}
\emph{where } $\Theta \left( x \right) = \int_0^x {t f_{\Pr}\left( t \right)dt}$,
$\Xi \left( x \right) = \int_0^x {\left( {1 - {f_{\Pr }}\left( t \right)} \right)tdt}$, $\varphi_\mathrm{LoS}\left(x\right)=
\left( {\frac{{{\beta_\mathrm{NLoS}}}}{{{\beta_\mathrm{LoS}}}}} \right)^{1/\alpha_\mathrm{NLoS}}
x^{{\alpha_\mathrm{LoS}}/{\alpha_\mathrm{NLoS}}}$, and $\varphi_\mathrm{NLoS}\left(x\right)=
\left( {\frac{{{\beta_\mathrm{LoS}}}}{{{\beta_\mathrm{NLoS}}}}} \right)^{1/\alpha_\mathrm{LoS}}
x^{{\alpha_\mathrm{NLoS}}/{\alpha_\mathrm{LoS}}}$.
\begin{proof}
According to \eqref{receive_power},  the average receive power can be written as
\begin{align}\label{Proof_exact_1}
 {\bar P}_\mathrm{r_o}=\mathbb{E}\left\{ {P_{\mathrm{r}_o}}\right\}=\mathbb{E}\left\{{\mathrm{En}_1}\right\}+\mathbb{E}\left\{{\mathrm{En}_2}\right\}.
\end{align}
Note that without considering the average inter-cell interfering power $\mathbb{E}\left\{{\mathrm{En}_2}\right\}$, $\mathbb{E}\left\{{\mathrm{En}_1}\right\}$ can be thought of as  the lower bound of $ {\bar P}_\mathrm{r_o}$, hence we first calculate the lower bound $ {\bar P}_\mathrm{r_o}^\mathrm{Low}=\mathbb{E}\left\{{\mathrm{En}_1}\right\}$. Due to the fact that a typical user can be connected to either a LoS BS or a NLoS BS,  using the law of total expectation, the lower bound for the average receive power is calculated as
\begin{align}\label{low_power_der1}
  {\overline P}_\mathrm{r_o}^\mathrm{Low}&=\mathbb{E}\left\{{\mathrm{En}_1}\right\}\nonumber\\
 &=\Lambda_\mathrm{LoS}\mathbb{E}\left\{{\mathrm{En}_1|\mathrm{LoS}}\right\}+
\Lambda_\mathrm{NLoS}\mathbb{E}\left\{{\mathrm{En}_1|\mathrm{NLoS}}\right\},
\end{align}
where $\Lambda_\mathrm{LoS}$ represents the probability that the typical user is connected to a LoS BS and $\Lambda_\mathrm{NLoS}=1-\Lambda_\mathrm{LoS}$ represents the probability that the typical user is connected to a NLoS BS.
We first derive $\mathbb{E}\left\{{\mathrm{En}_1|\mathrm{LoS}}\right\}$ as
\begin{align}\label{low_power_der2}
&\mathbb{E}\left\{{\mathrm{En}_1|\mathrm{LoS}}\right\}=\mathbb{E}\left\{ NM{P_{\mathrm{mm}}}L\left(\max\left\{  {{R_o}}, \mathrm{D}\right\}\right)|\mathrm{LoS}\right\}\nonumber\\
&=NM{P_{\mathrm{mm}}}\int_0^\infty L\left(\max\left\{  {x}, \mathrm{D}\right\}\right) f_R^{\mathrm{LoS}}\left( x \right){dx}\nonumber\\
& =NM{P_{\mathrm{mm}}}\beta_\mathrm{LoS}\Big({\rm{D}}^{-\alpha_\mathrm{LoS}} \int_0^{\rm{D}} f_R^{\mathrm{LoS}}\left( x \right){dx} \nonumber\\
&\hspace{0.4 cm}+\int_{\rm{D}}^{\infty} x^{-\alpha_\mathrm{LoS}} f_R^{\mathrm{LoS}}\left( x \right){dx}  \Big),
\end{align}
where $f_R^{\mathrm{LoS}}\left( x \right)$ is the conditional
probability density function (PDF) of the distance $R$ between the user and the serving LoS BS given that the user is connected to a LoS BS, which is given by \cite[Lemma 3]{Tianyang_arxiv2014}
\begin{align}\label{LoS_Distance_conditionPDF}
f_R^{\mathrm{LoS}}\left( x \right) = \frac{{2\pi \rho }}{{{\Lambda _\mathrm{LoS}}}}x{f_{\Pr }}
\left( x \right){e^{ - 2\pi \rho \left[ {\Theta \left( x \right) + \Xi \left( {{\varphi_\mathrm{LoS}}
\left( x \right)} \right)} \right]}}.
\end{align}
We then derive $\mathbb{E}\left\{{\mathrm{En}_1|\mathrm{NLoS}}\right\}$ as
\begin{align}\label{En1_NLOS_LOW}
&\mathbb{E}\left\{{\mathrm{En}_1|\mathrm{NLoS}}\right\}=\mathbb{E}\left\{ NM{P_{\mathrm{mm}}}L\left(\max\left\{  {{R_o}}, \mathrm{D}\right\}\right)|\mathrm{NLoS}\right\}\nonumber\\
&=NM{P_{\mathrm{mm}}}\int_0^\infty L\left(\max\left\{  {x}, \mathrm{D}\right\}\right) f_R^{\mathrm{NLoS}}\left( x \right){dx}\nonumber\\
& =NM{P_{\mathrm{mm}}}\beta_\mathrm{NLoS}\Big({\rm{D}}^{-\alpha_\mathrm{NLoS}} \int_0^{\rm{D}} f_R^{\mathrm{NLoS}}\left( x \right){dx} \nonumber\\
&\hspace{0.4 cm}+ \int_{\rm{D}}^{\infty} x^{-\alpha_\mathrm{NLoS}} f_R^{\mathrm{NLoS}}\left( x \right){dx}\Big) \bigg\},
\end{align}
where
 $f_R^{\mathrm{NLoS}}\left( x \right)$ is the conditional
PDF of the distance $R$ between the user and the serving NLoS BS given that the user is connected to a NLoS BS, which is given by~\cite[Lemma 3]{Tianyang_arxiv2014}
\begin{align}\label{NLoS_Distance_conditionPDF}
\hspace{-0.4cm} f_R^{\mathrm{NLoS}}\left( x \right) = \frac{{2\pi \rho }}{{{\Lambda _\mathrm{NLoS}}}}x{\left(1-f_{\Pr }\left( x \right)\right)}
{e^{ - 2\pi \rho \left[ \Theta \left(\varphi_\mathrm{NLoS}\left(x\right)\right)+ {\Xi \left( x\right)} \right]}}.
\end{align}

Substituting  \eqref{low_power_der2} and  \eqref{En1_NLOS_LOW} into \eqref{low_power_der1}, we get the lower bound expression for the average receive power.

Similar to \eqref{low_power_der1}, $\mathbb{E}\left\{{\mathrm{En}_2}\right\}$ in \eqref{Proof_exact_1} is derived as
\begin{align}\label{Sec_Part_Exp_En}
\hspace{-0.3 cm} \mathbb{E}\left\{{\mathrm{En}_2}\right\}=\Lambda_\mathrm{LoS}\mathbb{E}\left\{{\mathrm{En}_2|\mathrm{LoS}}\right\}+
\Lambda_\mathrm{NLoS}\mathbb{E}\left\{{\mathrm{En}_2|\mathrm{NLoS}}\right\}.
\end{align}
As mentioned in \cite{Tianyang_arxiv2014}, we  assume that the correlations of the shadowing between the links are ignored. By using the thinning theorem of PPP~\cite{Baccelli2009}, the LOS BS process $\Phi_\mathrm{LoS}$ and the NLOS BS process $\Phi_\mathrm{NLoS}$
form two independent non-homogeneous PPPs with density functions $f_\mathrm{Pr}\left(R\right)\rho$ and $\left(1-f_\mathrm{Pr}\left(R\right)\right) \rho$
in polar coordinates, respectively. Therefore, $\mathbb{E}\left\{{\mathrm{En}_2|\mathrm{LoS}}\right\}$ can be divided into two independent components as follows:
\begin{align}\label{LT_1_2_Exp}
&\hspace{-0.2 cm}\mathbb{E}\left\{{\mathrm{En}_2|\mathrm{LoS}}\right\}\nonumber\\
&\hspace{-0.2 cm}=\frac{{{P_{\mathrm{mm}}}}}{{NM}}\mathbb{E}\left\{\sum\nolimits_{k \in \Phi_\mathrm{LoS} \setminus \left\{{B\left(o;R_o\right) }\right\}} {\hbar _k} L\left( \max\left\{{{{R_k}}}, \mathrm{D}\right\}\right)\right\}\nonumber\\
&\hspace{-0.2 cm}+\frac{{{P_{\mathrm{mm}}}}}{{NM}}\mathbb{E}\left\{\sum\nolimits_{k \in \Phi_\mathrm{NLoS} \setminus \left\{{B\left(o;\varphi_\mathrm{LoS}\left(R_o\right)\right) }\right\} } {\hbar _k} L\left( \max\left\{{{{R_k}}}, \mathrm{D}\right\}\right)\right\},
\end{align}
where $B\left(o;x\right)$ denotes ball of radius $x$ centered at the origin $o$. Using Campbell's theorem~\cite{Baccelli2009}, \eqref{LT_1_2_Exp} can be derived as
\begin{align}\label{En_2_equation_Exp}
&\mathbb{E}\left\{{\mathrm{En}_2|\mathrm{LoS}}\right\}\nonumber\\
&=\frac{{{P_{\mathrm{mm}}}}}{{NM}}\overline{\hbar} \beta_\mathrm{LoS}\times\nonumber\\
&\hspace{-0.4 cm}\int_{0}^{\infty}{ {\left[\int\nolimits_{{\mathbb{R}^{\rm{2}}} \setminus B\left(o;x\right)}
{ \left(\max\left\{t,\mathrm{D}\right\}\right)^{-\alpha_\mathrm{LoS}}}\mathbb{M}\left(dt\right)\right] }f_R^{\mathrm{LoS}}\left( x \right) dx}\nonumber\\
&+\frac{{{P_{\mathrm{mm}}}}}{{NM}}\overline{\hbar} \beta_\mathrm{NLoS}\times\nonumber\\
&\hspace{-0.7 cm}\int_{0}^{\infty}{ {\left[\int\nolimits_{{\mathbb{R}^{\rm{2}}} \setminus {B\left(o;\varphi_\mathrm{LoS}\left(x\right)\right) }}
{ \left(\max\left\{t,\mathrm{D}\right\}\right)^{-\alpha_\mathrm{NLoS}}}\mathbb{M}\left(dt\right)\right] }f_R^{\mathrm{LoS}}\left( x \right) dx}
\\
& \hspace{-0.6 cm}\mathop  = \limits^{\left(a\right)} \frac{{{P_{\mathrm{mm}}}}}{{NM}}\overline{\hbar} \beta_\mathrm{LoS}2 \pi \rho \times \nonumber\\
&\hspace{-0.6 cm}\int_{0}^{\infty}{ {\left[\int_{x}^{\infty}
{ \left(\max\left\{t,\mathrm{D}\right\}\right)^{-\alpha_\mathrm{LoS}}}f_\mathrm{Pr}\left(t\right) t dt \right] }f_R^{\mathrm{LoS}}\left( x \right) dx}\nonumber\\
&\hspace{-0.6 cm}+\frac{{{P_{\mathrm{mm}}}}}{{NM}}\overline{\hbar} \beta_\mathrm{NLoS}
2 \pi \rho\times \nonumber\\
&\hspace{-0.6 cm}\int_{0}^{\infty}{ {\left[\int_{\varphi_\mathrm{LoS}\left(x\right)}^{\infty}
{ \left(\max\left\{t,\mathrm{D}\right\}\right)^{-\alpha_\mathrm{NLoS}}} \left(1-f_\mathrm{Pr}\left(t\right)\right) t dt\right] }f_R^{\mathrm{LoS}}\left( x \right) dx},
\end{align}
where  $\mathbb{M}\left(\mathcal{A}\right)$ in (14) represents the mean number of points of a spatial point process in $\mathcal{A}$~\cite{Baccelli2009},  $\left(a\right)$ results from using the polar-coordinate system and $\overline{\hbar}$  is given by
\begin{align}\label{h_bar_1}
&\overline{\hbar}=\mathbb{E}\left\{ \hbar _k\right\}=\mathbb{E}\left\{\frac{{\left( {1 - \cos \left( {N{\varpi _{{\vartheta _{{\mathrm{r}_k}}},{\theta _{{\mathrm{r}_o}}}}}}
\right)} \right)}}{\left( {1 - \cos \left( {{\varpi _{{\vartheta _{{\mathrm{r}_k}}},{\theta _{{\mathrm{r}_o}}}}}} \right)} \right)}\right\}\nonumber\\
&\times \mathbb{E}\left\{\frac{\left( {1 - \cos \left( {M{\varpi _{{\vartheta _{{\mathrm{t}_k}}},
{\theta _{{\mathrm{t}_k}}}}}} \right)} \right)}
{{\left( {1 - \cos \left( {{\varpi _{{\vartheta _{{\mathrm{t}_k}}},
{\theta _{{\mathrm{t}_k}}}}}} \right)} \right)}}\right\}\nonumber\\
&=\int_{\rm{0}}^{{\rm{2}}\pi } {\int_0^{2\pi } {\frac{{\left( {1 - \cos \left( {N{\varpi _{\vartheta _{{\mathrm{r}_k}},{\theta _{{\mathrm{r}_o}}}}}} \right)} \right)}}
{{\left( {1 - \cos \left( {{\varpi _{\vartheta _{{\mathrm{r}_k}},{\theta _{{\mathrm{r}_o}}}}}} \right)} \right)}}\frac{{\vartheta _{{\mathrm{r}_k}}{\theta _{{\mathrm{r}_o}}}}}{{4{\pi ^2}}}d{\vartheta _{{\mathrm{r}_k}}}d{{\theta _{{\mathrm{r}_o}}}}} } \times \nonumber\\
& \int_{\rm{0}}^{{\rm{2}}\pi } {\int_0^{2\pi } {\frac{{\left( {1 - \cos \left( {M{\varpi _{{\vartheta _{{\mathrm{t}_k}}},{\theta _{{\mathrm{t}_k}}}}}} \right)} \right)}}
{{\left( {1 - \cos \left( {{\varpi _{{\vartheta _{{\mathrm{t}_k}}},{\theta _{{\mathrm{t}_k}}}}}} \right)} \right)}}\frac{{{\vartheta _{{\mathrm{t}_k}}}{\theta _{{\mathrm{t}_k}}}}}{{4{\pi ^2}}}d{{\vartheta _{{\mathrm{t}_k}}}}d{\theta _{{\mathrm{t}_k}}}} }.
\end{align}

Similar to \eqref{LT_1_2_Exp}, $\mathbb{E}\left\{{\mathrm{En}_2|\mathrm{NLoS}}\right\}$ is given by
\begin{align} \label{En_2_NLOS_Exp}
&\mathbb{E}\left\{{\mathrm{En}_2|\mathrm{NLoS}}\right\} \nonumber\\
&=\frac{{{P_{\mathrm{mm}}}}}{{NM}} \overline{\hbar}\mathbb{E}\left\{\sum\nolimits_{k \in \Phi_\mathrm{LoS} \setminus \left\{{B\left(o;\varphi_\mathrm{NLoS}\left(R_o\right)\right) }\right\}} L\left( \max\left\{{{{R_k}}}, \mathrm{D}\right\}\right)\right\}\nonumber\\
&+\frac{{{P_{\mathrm{mm}}}}}{{NM}}\overline{\hbar}\mathbb{E}\left\{\sum\nolimits_{k \in \Phi_\mathrm{NLoS} \setminus \left\{{B\left(o;R_o\right) }\right\} } L\left( \max\left\{{{{R_k}}}, \mathrm{D}\right\}\right)\right\}\nonumber\\
&=\frac{{{P_{\mathrm{mm}}}}}{{NM}}\overline{\hbar} \beta_\mathrm{LoS}
2 \pi \rho \times \nonumber\\
&\int_{0}^{\infty}{ {\left[\int_{\varphi_\mathrm{NLoS}\left(x\right)}^{\infty}
{ \left(\max\left\{t,\mathrm{D}\right\}\right)^{-\alpha_\mathrm{LoS}}}f_\mathrm{Pr}\left(t\right) t dt \right] }f_R^{\mathrm{NLoS}}\left( x \right) dx}\nonumber\\
&+\frac{{{P_{\mathrm{mm}}}}}{{NM}}\overline{\hbar} \beta_\mathrm{NLoS}
2 \pi \rho \times \nonumber \\
&\int_{0}^{\infty}{ {\left[\int_{x}^{\infty}
{ \left(\max\left\{t,\mathrm{D}\right\}\right)^{-\alpha_\mathrm{NLoS}}} \left(1-f_\mathrm{Pr}\left(t\right)\right) t dt\right] }f_R^{\mathrm{NLoS}}\left( x \right) dx}.
\end{align}
Substituting (16) and \eqref{En_2_NLOS_Exp} into \eqref{Sec_Part_Exp_En}, we get $\mathbb{E}\left\{{\mathrm{En}_2}\right\}$. Substituting $\mathbb{E}\left\{{\mathrm{En}_1}\right\}$ in \eqref{low_power_der1}
and $\mathbb{E}\left\{{\mathrm{En}_2}\right\}$ into \eqref{Proof_exact_1}, we obtain the desired result given by \eqref{average_power_transfer_1}.
\end{proof}
\end{theorem}

\section{Performance Evaluation}
In this section, we examine the average achievable rate for the uplink. Based on Section \ref{mmWave_WPT}, each user sets its transmit power to a stable value $P_\mathrm{u}=\eta\frac{\phi}{\left(1-\phi\right)}{\overline P}_\mathrm{r_o}$, where the average receive power ${\overline P}_\mathrm{r_o}$ is given by \textbf{Theorem} 1.

The average achievable rate can be expressed as
\begin{align}\label{throughput_exp}
\mathbb{R}={\left(1-\phi\right)}\frac{1}{{\ln 2}}\int_0^\infty  {\frac{1 - {F_{\mathrm{SINR}}}\left( x \right)}{{1+x}}} dx,
\end{align}
where ${F_{\mathrm{SINR}}}\left( x \right)$ is the CDF of the receive SINR at a typical serving BS. Since a simple expression for ${F_{\mathrm{SINR}}}\left( x \right)$ is intractable, we derive an upper bound of the average achievable rate as
\begin{align}\label{throughput_exp_upper}
\mathbb{R}^{\mathrm{Upper}}={\left(1-\phi\right)}\frac{1}{{\ln 2}}\int_0^\infty  {\frac{{1 - {F_{\mathrm{SNR}}}\left( x \right)}}{{1+x}}} dx,
\end{align}
where ${F_{\mathrm{SNR}}}$ is the CDF of the receive SNR at a typical serving BS without mmWave inter-cell uplink interference, which is presented in the following theorem.
\begin{theorem}
\emph{The CDF expression for the receive SNR at a typical serving BS is given by}
\begin{align}\label{upper_CDF_Uplink}
&\hspace{-0.6cm} {F_{\mathrm{SNR}}}\left( x \right)={\rm{\mathbf{1}}}\left(\mathrm{D}>\Delta_1\right){{2\pi \rho }} \int_0^\mathrm{D} t{f_{\Pr }}
\left( t \right){e^{ - 2\pi \rho \left[ {\Theta \left( t \right) + \Xi \left( {{\varphi_\mathrm{LoS}}
\left( t \right)} \right)} \right]}}  {dt}\nonumber\\
&\hspace{-0.6cm}+{{2\pi \rho }} \int_{\max\left\{\mathrm{D},\Delta_1\right\}}^\infty t{f_{\Pr }}
\left( t \right){e^{ - 2\pi \rho \left[ {\Theta \left( t \right) + \Xi \left( {{\varphi_\mathrm{LoS}}
\left( t \right)} \right)} \right]}}  {dt}\nonumber\\
&\hspace{-0.6cm}+{\rm{\mathbf{1}}}\left(\mathrm{D}>\Delta_2\right){{2\pi \rho }}\int_0^\mathrm{D} t{\left(1-f_{\Pr }\left( t \right)\right)}
{e^{ - 2\pi \rho \left[ \Theta \left(\varphi_\mathrm{NLoS}\left(t\right)\right)+ {\Xi \left( t\right)} \right]}} {dt}\nonumber\\
&\hspace{-0.6cm}+{{2\pi \rho }}\int_{\max\left\{\mathrm{D},\Delta_2\right\}}^\infty t{\left(1-f_{\Pr }\left( t \right)\right)}
{e^{ - 2\pi \rho \left[ \Theta \left(\varphi_\mathrm{NLoS}\left(t\right)\right)+ {\Xi \left( t\right)} \right]}} {dt},
\end{align}
\emph{where ${\rm{\mathbf{1}}}\left(A\right)$ is the indicator function that returns one if the condition $A$ is satisfied, $P_{\mathrm{u}}=\eta\frac{\phi}{\left(1-\phi\right)}{\overline P}_\mathrm{r_o}$ with ${\overline P}_\mathrm{r_o}$ given by \eqref{average_power_transfer_1}, $\Delta_1= \left(\frac{MN{P_{\mathrm{u}}}\beta_\mathrm{LoS}}{x \delta^2}\right)^{1/\alpha_\mathrm{LoS}}$, and $\Delta_2=\left(\frac{MN{P_{\mathrm{u}}}\beta_\mathrm{NLoS}}{x \delta^2}\right)^{1/\alpha_\mathrm{NLoS}}$.}
\end{theorem}
\begin{proof}
Based on \eqref{Eq_1_inform}, the CDF of the receive SNR at a typical serving BS is defined as
\begin{align}\label{upper_CDF_Uplink_1}
&{F_{\mathrm{SNR}}}\left( x \right)=\Pr\left( \mathrm{SNR} < x \right)\nonumber\\
&=\Pr\left(\frac{MN{P_{\mathrm{u}}}L\left( \max\left\{  {{R_o}}, \mathrm{D}\right\} \right)}{\delta^2}<x\right).
\end{align}
Due to blockage effects, the user is connected to either a LoS BS or a NLoS BS. From the law of total probability, \eqref{upper_CDF_Uplink_1} can be reexpressed as
\begin{align}\label{upper_CDF_Uplink_2_step1}
{F_{\mathrm{SNR}}}\left( x \right)
&=\Lambda_\mathrm{LoS}\Pr\left(\frac{MN{P_{\mathrm{u}}}L\left( \max \left\{  {{R_o}}, \mathrm{D}\right\}\right)}{\delta^2}<x|\mathrm{LoS} \right)\nonumber\\
&+\Lambda_\mathrm{NLoS}\Pr\left(\frac{MN{P_{\mathrm{u}}}L\left( \max \left\{  {{R_o}}, \mathrm{D}\right\}\right)}{\delta^2}<x|\mathrm{NLoS} \right)\nonumber\\
&\mathop = \limits^{\left( a \right)} \Lambda_\mathrm{LoS} {\rm{\mathbf{1}}}\left(\mathrm{D}>\Delta_1\right) \Pr\left(R_o < \mathrm{D}|\mathrm{LoS} \right)\nonumber\\
&+\Lambda_\mathrm{LoS}\Pr\left(R_o>\max\left\{\mathrm{D},\Delta_1\right\}|\mathrm{LoS} \right)\nonumber\\
&+\Lambda_\mathrm{NLoS}   {\rm{\mathbf{1}}}\left(\mathrm{D}>\Delta_2\right) \Pr\left(R_o < \mathrm{D}|\mathrm{NLoS} \right)\nonumber\\
&+\Lambda_\mathrm{NLoS} \Pr\left(R_o>\max\left\{\mathrm{D},\Delta_2\right\}|\mathrm{NLoS} \right)
\end{align}
where (a) follows from the fact that the distance between the typical user and the typical serving BS should be larger than a minimum value $\Delta$ such that the receive SNR drops below a threshold $x$. For LoS, $\Delta=\Delta_1$, and for NLoS, $\Delta=\Delta_2$. As such, we can further calculate \eqref{upper_CDF_Uplink_2_step1} as
\begin{align}\label{upper_CDF_Uplink_2}
{F_{\mathrm{SNR}}}\left( x \right)&=\Lambda_\mathrm{LoS}  {\rm{\mathbf{1}}}\left(\mathrm{D}>\Delta_1\right)\int_0^{\mathrm{D}} f_R^{\mathrm{LoS}}\left( t \right) {dt} \nonumber\\
&+\Lambda_\mathrm{LoS} \int_{\max\left\{\mathrm{D},\Delta_1\right\}}^\infty f_R^{\mathrm{LoS}}\left( t \right) {dt}\nonumber\\
&+\Lambda_\mathrm{NLoS} {\rm{\mathbf{1}}}\left(\mathrm{D}>\Delta_2\right)\int_0^{\mathrm{D}} f_R^{\mathrm{NLoS}}\left( t \right) {dt} \nonumber\\
&+\Lambda_\mathrm{NLoS} \int_{\max\left\{\mathrm{D},\Delta_2\right\}}^\infty f_R^{\mathrm{NLoS}}\left( t \right) {dt}.
\end{align}

Substituting $f_R^{\mathrm{LoS}}\left( t \right)$ in \eqref{LoS_Distance_conditionPDF} and
$ f_R^{\mathrm{NLoS}}\left( t \right)$ in \eqref{NLoS_Distance_conditionPDF} into \eqref{upper_CDF_Uplink_2}, we obtain the desired result in \eqref{upper_CDF_Uplink}.
\end{proof}

By substituting \eqref{upper_CDF_Uplink} into \eqref{throughput_exp_upper}, the upper bound expression for the average achievable rate is obtained.

\section{Numerical Examples}
We provide numerical examples to understand the impact of the BS density and number of antennas on the harvested energy and achievable rate. The mmWave network is assumed to operate at 38 GHz, the mmWave BS transmit power is $P_{\mathrm{mm}}=43$ dBm, $\alpha_\mathrm{LoS}=2$, $\alpha_\mathrm{NLoS}=4$, $d=\frac{\lambda}{2}$, $\phi =0.5$ and  $\eta=0.5$. Let the LoS probability function be $f_\mathrm{Pr}\left(R\right)=e^{-\varrho R}$ with $\varrho=141.4$ meters~\cite{Tianyang_arxiv2014}. The mmWave bandwidth is $\mathrm{BW}=2$ GHz, the noise figure is $\mathrm{Nf}=10$ dB, the noise power $\delta^2=-174+10\log10\left(\mathrm{BW}\right)+\mathrm{Nf}$ dBm, and the reference distance $\mathrm{D}=1$. Each user is equipped with $N=16$ antennas.

\begin{figure}[!htp]
    \begin{center}
        \includegraphics[width=3.5in,height=3.25in]{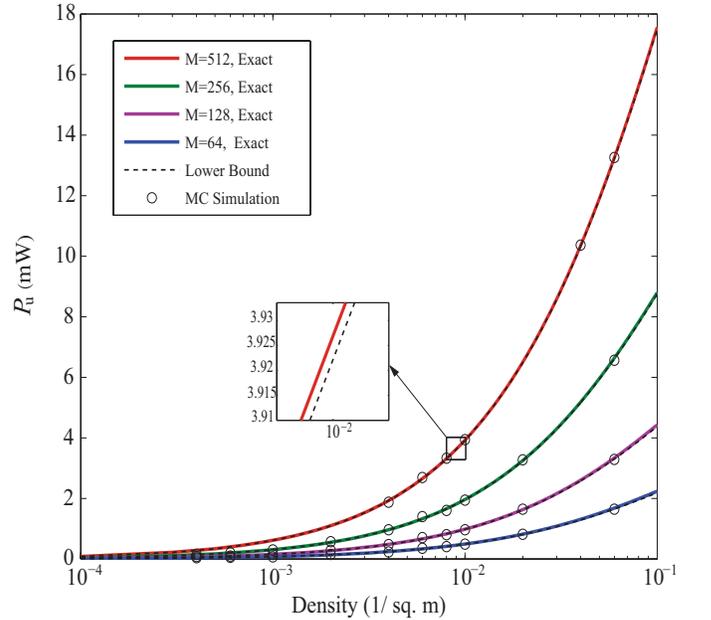}
        \caption{ The maximum continuous transmit power versus density with different $M$.}
        \label{fig:1}
    \end{center}
\end{figure}
Fig. \ref{fig:1} plots the maximum continuous transmit power $P_\mathrm{u}=\eta\frac{\phi}{\left(1-\phi\right)}
\mathbb{E}\left\{  P_\mathrm{r_o}\right\}$ versus BS density $\rho$ with different numbers of BS's antennas $M$. The exact curves are obtained based on \eqref{average_power_transfer_1}, and their lower bounds are obtained based on \eqref{average_transmit_power}. We first see that the exact curves have a precise match with the Monte Carlo (MC) simulations, which validate our theoretical analysis. We next see that the lower bounds tightly match with the corresponding exact curves, which indicates that in mmWave networks, the interference has little impact on the harvested power.  Furthermore, the use of large antenna arrays can effectively improve the harvested energy.

\begin{figure}[!htp]
    \begin{center}
        \includegraphics[width=3.5in,height=3.25in]{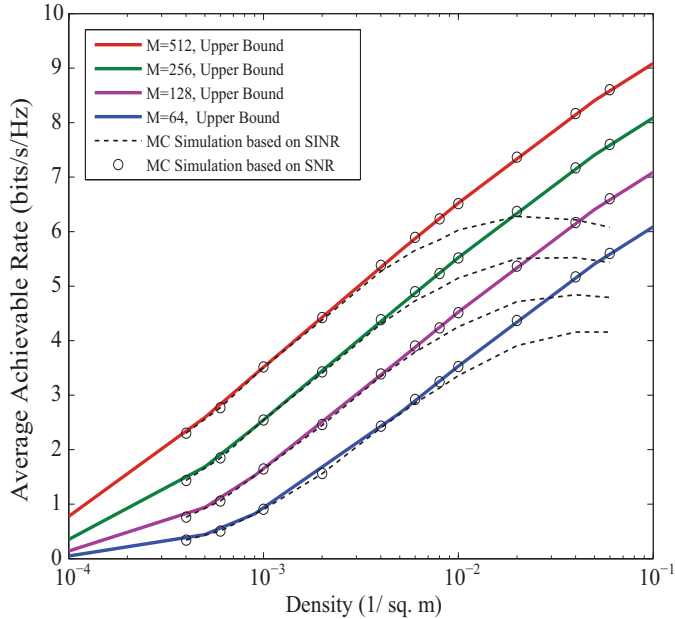}
        \caption{ The average achievable rate versus density with different $M$.}
        \label{fig:2}
    \end{center}
\end{figure}

Fig. \ref{fig:2} plots the average achievable rate  versus density  with different numbers of BS's antennas $M$. The solid curves obtained from \eqref{throughput_exp_upper} are the upper bound of the average achievable rate, and have a good match with MC simulations marked with $\circ$. The dash lines obtained from the MC simulations are the exact average achievable rate. We first see that increasing the number of antennas  brings additional array gains and improves the achievable rate. The mmWave transmission with large antenna arrays can achieve enormous throughput with the help of large bandwidth and large array gains. We also see that
the mmWave networks are noise-limited when BSs are not super dense, however, increasing BS density beyond a critical point, the mmWave networks will switch to be interference-limited, and the average achievable rate decreases with increasing BS density. Similar conclusions have also been mentioned in~\cite{Tianyang_arxiv2014}, where downlink performance is examined without considering
energy harvesting.

\section{Conclusion}

In this paper, the millimeter wave networks with wireless power transfer was took into account. Before uplink transmission, users harvested energy from its serving base station and the interfering base stations. We first derived the average harvested energy for the case of user with large energy storage, to examine the amount of  power transferred by  millimeter wave base stations. We then derived the average achievable uplink rate, to examine the performance of uplink transmission using the harvested energy. The results provide important insights into the design and application of the millimeter wave power transfer.


\end{document}